\documentclass[12pt]{article}
\usepackage[letterpaper,top=2.5cm,bottom=2cm,left=3cm,right=3cm,marginparwidth=1.75cm]{geometry}
\usepackage{graphicx,amsfonts,amsmath,tikz-cd,amsthm}
\usepackage[final]{hyperref}
\usepackage{amssymb}
\usepackage{authblk}
\usepackage{tikz}
\usepackage{subfigure}
\usepackage[normalem]{ulem}
\usepackage{algpseudocode}
\usepackage[ruled,vlined,linesnumbered]{algorithm2e}
\usepackage{enumitem}

\usepackage[capitalise, noabbrev, nameinlink]{cleveref}

\theoremstyle{theorem}
\newtheorem{theorem}{Theorem}[section] 
\newtheorem{proposition}[theorem]{Proposition}
\newtheorem{observation}[theorem]{Observation}

\theoremstyle{theorem}
\newtheorem{definition}[theorem]{Definition} 
\newtheorem{example}[theorem]{Example}

\theoremstyle{plain}
\newtheorem{remark}[theorem]{Remark}
\newtheorem*{introtheorem}{Theorem}

\title{The geometry of higher-order modern portfolio theory}
\author{Emil Horobe\c{t}}
\affil{ {\small Sapientia Hungarian University of Transylvania, Romania}}
\date{}

\begin{document}

\maketitle
\begin{abstract}
In this article, we study the generalized modern portfolio theory, with utility functions admitting higher-order cumulants. We establish that under certain genericity conditions, the utility function has a constant number of complex critical points. We study the discriminant locus of complex critical points with multiplicity. Finally, we switch our attention to the generalization of the feasible portfolio set (variety), determine its dimension, and give a formula for its degree.
\end{abstract}
\section{Introduction}
Classical portfolio theory, as formalized by Markowitz \cite{Mar}, rests on the assumption that risk can be adequately described by the variance of returns and that investors’ preferences can be captured by quadratic utility. While this framework has proven foundational, it is increasingly recognized as restrictive: financial return distributions often exhibit skewness, kurtosis, and higher-order irregularities that cannot be explained by second-order moments alone. The resulting misalignment between model assumptions and empirical distributions motivates the development of higher-order portfolio theory.

In higher-order modern portfolio theory, we optimize the truncated Taylor expansion of expected utility, accounting for higher-order moments; more fundamentally, higher-order cumulants (see \cite{HS00}). This generalization provides a richer characterization of investors’ attitudes toward asymmetry and tail behavior, embedding skewness and kurtosis aversion (or preference) directly into the optimization framework. Recent works on this topic include \cite{Leon, Madal, Zhao}.

The challenge, however, is not simply technical (of incorporating additional terms into an expansion) but conceptual: the geometry of parameter selection within the utility function must be mathematically understood. Investors’ preferences for higher-order risks translate into weightings on cumulants that determine the shape of the feasible portfolio set. These sets are no longer ellipsoids, as in the quadratic case, but algebraic varieties whose degree grows rapidly with the order of cumulants used, whose geometry depends on both the order and relative scaling of chosen parameters. Understanding this geometry is crucial, for it dictates feasible diversification strategies, the stability of critical portfolios, and the interpretability of preference parameters.

The present work develops a geometric perspective on higher-order modern portfolio theory, first by understanding the generic number of critical portfolios and their discriminant locus. 

The \textbf{key feature} of this paper is that numerical solvers get stuck in local minima in these generalized non-convex landscapes, while our algebraic study characterizes the entire critical locus, ensuring that we can distinguish local minima from the global minimum. This is a clear financial way to interpret our main results, which are the following.

\begin{introtheorem}[A]
Under the assumptions of Theorem~\ref{numberofcriticalpoints} the generalized utility function \[L(x_1,x_2,\ldots,x_n)=\sum_{j=1}^d\sum_{i=1}^n w_jk_{ij}x_i^j\] 
has $(d-1)^{n-1}$ complex non-degenerate critical points.
\end{introtheorem}
\begin{introtheorem}[B]
A critical point, $(x_1,x_2,\ldots,x_n)$, of the utility function $L$ has a multiplicity if and only if it satisfies
\begin{equation}\label{discriminant}
    \sum_{i=1}^n \prod_{j\neq i}\frac{\partial P_j(x_j)}{\partial x_j}=0,
\end{equation}
where $P_i(x_i)=\sum_{j=1}^d j w_j k_{ij} x_i^{j-1}$.
\end{introtheorem}

Finally, we study the basic geometric invariants (dimension and degree) of the generalized feasible portfolio set (variety). Our main findings can be summarized as follows.
\begin{introtheorem}[C]
Under the conditions of Theorem~\ref{degree_of_Pk} the feasible portfolio variety $\mathcal{P}_k$ has dimension $n-1$ and its degree equals \[d\cdot(d-1)\cdot\ldots\cdot(d-n+2).\]
\end{introtheorem}

\section{Mathematical setup}
In this paper, we consider $X_1,X_2,\ldots,X_n$ independent random variables, where $X_i$ represents the relative (annual/quarterly/etc.) return of the $i^{th}$ asset in our portfolio. In this context, a portfolio $P$ is a random variable that is a convex combination of the $X_i$'s, that is
\[P=x_1X_1+x_2X_2+\ldots+x_nX_n,\] with $x_i\in (0,1)$ and $x_1+x_2+\ldots+x_n=1$.

We consider a utility function $L$ dependent on the portfolio $P$ through $x_1,x_2,\ldots,x_n$. Classically, $L$ depends on the expected value of $P$ and its variance. So
\[L(x_1,x_2,\ldots,x_n)=w_1 \mathbb{E}(P)+w_2 \mathbb{V}(P),\] where $w_1, w_2$ are weights and $\mathbb{E}(P)$ is the expected value of $P$ , and $\mathbb{V}(P)$ is its variance. 
Now, since we assumed that the $X_i$'s are independent, we have the additivity and homogeneity of the expected value, and of the variance, so we get
\[L(x_1,x_2\ldots,x_n)=w_1 \mathbb{E}(x_1X_1+x_2X_2+\ldots +x_nX_n)+w_2 \mathbb{V}(x_1X_1+x_2X_2+\ldots +x_nX_n)=\]
\[= w_1\sum_{i=1}^nx_i \mathbb{E}(X_i)+w_2\sum_{i=1}^n x_i^2\mathbb{V}(X_i).\]
For almost all choices of weights, $w_1,w_2,$ this is a quadratic polynomial function that can be easily optimized relative to the constraints $x_1+x_2=1$ and $x_1,x_2\in (0,1)$.
\begin{example}[Markowitz utility on a two asset portfolio]
Let us now consider a simple portfolio consisting of two independent assets and their corresponding $X_1$ and $X_2$ independent returns (we assume that the expected values and variances of these variables are non-zero). We can write the portfolio as
\[P=x\cdot X_1+(1-x)\cdot X_2,\] with $x\in (0,1)$.
The classical Markowitz utility function equals
\[L(x)=\mathbb{E}(P)-\frac{1}{2}\mathbb{V}(P)=x\cdot \mathbb{E}(X_1)+(1-x)\cdot\mathbb{E}(X_2)-\frac{1}{2}\left(x^2\cdot\mathbb{V}(X_1)+(1-x)^2\cdot\mathbb{V}(X_2)\right).\]
Which is a quadratic function in $x$, always attaining a maximum because the leading coefficient $\left(-\frac{1}{2}\right)(\mathbb{V}(X_1)+\mathbb{V}(X_2))$ is always negative. And this maximum is attained for
\[x^*=
\frac{\mathbb{E}(X_1)-\mathbb{E}(X_2)+\mathbb{V}(X_2)}{\mathbb{V}(X_1)+\mathbb{V}(X_2)},\] if this value is in between $(0,1)$. If this value is negative, then we set $x^*=0$, and as a consequence, $X_1$ disappears from the portfolio. Or if this value is greater than one, then we set $x^*=1$ and, as a consequence, $X_2$ disappears from the portfolio. In such cases, we say that we have gotten a degenerate portfolio. 
\end{example}
In this paper, from now on, we only consider \textbf{non-degenerate portfolios}, meaning that we assume that $x_i\neq 0$, for all $i=1,\ldots,n$.
\subsection{Higher-order utility functions}
We have seen in the classical theory that it is crucial that for independent variables $X_1, X_2,\ldots,X_n$ the expected value and the variance utilized in $L$ are additive and homogeneous. In this fashion, we propose to generalize the theory by using higher-order cumulants.
\begin{definition}
The cumulants of a random variable $X$ are defined using the cumulant-generating function $K(t)=\log \mathbb{E}(e^{tX})$. Now, if $\displaystyle{K(t)=\sum_{j=1}^{\infty}\kappa_j\frac{t^j}{j!}}$ is the power series expansion of $K(t)$, then the \textbf{$j$-th cumulant} of $X$ is equal to $\kappa_j$.
\end{definition}
For independent variables, we have the property that $\kappa_j$ are additive and homogeneous for any $j\geq 0$ and that the first cumulants are equal to the moments (expected value, variance, etc.).

Now, for generic weights, let us consider the following generalized utility function
\[L(x_1,x_2,\ldots,x_n)=\sum_{j=1}^d w_j\kappa_j(P),\]
where, by the properties of cumulants and because the $X_i$ are independent, we get that
\[\kappa_j(P)=\sum_{i=1}^n x_i^j \kappa_j(X_i).\]
So we want to optimize 
\[L(x_1,x_2,\ldots,x_n)=\sum_{j=1}^d\sum_{i=1}^n w_jx_i^j\kappa_{j}(X_i),\] with respect to $\sum_{i=1}^nx_i=1$ and $x_i\in(0,1)$, for all $i=1,\ldots,n$.
\section{The critical equations and the generic number of critical points}
We consider the complexified critical point system to study the algebraic geometry of the solutions, and we say that we want to find all complex critical points of the function  
\[L(x_1,x_2,\ldots,x_n)=\sum_{j=1}^d\sum_{i=1}^n w_jx_i^j\kappa_{j}(X_i),\] with respect to $\sum_{i=1}^nx_i=1$, where we consider $x_i\in \mathbb{C}$.
By doing this, we shed light on the geometry of the critical points variety (also, we get an understanding of the isolated Morse singularities of $L$). On the other hand, we lose the positivity of the solutions, so at the end, we will have to work on finding and understanding the real and positive solutions.
The first question to answer is whether the number of constrained critical points of the generalized utility function is constant for generic parameter data. We consider the cumulant values of the $X_i$'s to be fixed, real, and non-zero. We will denote $\kappa_j(X_i)$ by $k_{ij}$, and hence we get a fixed matrix of data, $K=(k_{ij})_{ij}$, comprising all the statistical data of the elements of the portfolio.

Now we use Lagrange multipliers to find the constrained critical points of $L$ with respect to $\sum_{i=1}^n x_i=1$. The Lagrange function of this problem is
\[\mathcal{L}(x_1,\ldots,x_n; \lambda)=\sum_{j=1}^d\sum_{i=1}^n w_jk_{ij}x_i^j-\lambda\left(\sum_{i=1}^n x_i-1\right).\]
Taking derivatives with respect to the $x_i$'s and to $\lambda$, we get the following system of critical equations
\begin{equation}\label{criticalequ}
 \begin{cases}
\sum_{j=1}^d j w_j k_{ij} x_i^{j-1}=\lambda,\text{ for }i=1,\ldots,n,\\
\sum_{i=1}^n x_i=1.
\end{cases}   
\end{equation}
The system consists of $n$ degree $d-1$ univariate polynomials and one linear condition in $n+1$ complex variables. So, in general, it has a zero-dimensional (meaning discrete) solution set, and we have the following result.
\begin{theorem}\label{numberofcriticalpoints}
Under the assumption that $k_{ij}$ are real, non-zero, and that $w_d\neq 0$ the generalized utility function \[L(x_1,x_2,\ldots,x_n)=\sum_{j=1}^d\sum_{i=1}^n w_jk_{ij}x_i^j\] 
generically has $(d-1)^{n-1}$ complex non-degenerate critical points with respect to the constraint $\sum_{i=1}^n x_i=1$. Here, non-degenerate means that no $x_i$ equals zero.
\end{theorem}
\begin{proof}
System \ref{criticalequ} can be rewritten in the following form
\begin{equation}\label{originalequations}
\begin{cases}
\sum_{j=1}^d j w_j k_{1j} x_1^{j-1}=\lambda,\\
\sum_{j=1}^d j w_j k_{1j} x_1^{j-1}=\sum_{j=1}^d j w_j k_{ij} x_i^{j-1},\text{ for }i=2,\ldots,n,\\
\sum_{i=1}^n x_i=1.
\end{cases}   
\end{equation}
The number of solutions of the above system equals the number of solutions of the system
\begin{equation}\label{rewrittencriticalequ}
 \begin{cases}
\sum_{j=1}^d j w_j k_{1j} x_1^{j-1}=\sum_{j=1}^d j w_j k_{ij} x_i^{j-1},\text{ for }i=2,\ldots,n,\\
\sum_{i=1}^n x_i=1.
\end{cases}   
\end{equation}
This is true because once we have a value for $x_1$, then $\lambda$ is uniquely defined by it in \ref{originalequations}.
Now, by Bézout's theorem, we get that the number of complex solutions to system \ref{rewrittencriticalequ} is at most the product of the degrees of these equations, which is exactly $(d-1)\cdot (d-1)\cdot\ldots\cdot (d-1)\cdot 1=(d-1)^{n-1}$. We need the assumption that $w_d\neq 0$ to ensure that the polynomials are actually of degree $d-1$ and not less.
Now we see that for the specific instance of $w_1=w_2=\ldots=w_{d-1}=0$ , and $w_d\neq 0$ we get
\[
\begin{cases}
    dw_dk_{1d}x_1^{d-1}=dw_dk_{id}x_i^{d-1},\text{ for }i=2,\ldots,n,\\
\sum_{i=1}^n x_i=1.
\end{cases}
\]
That is
\[
\begin{cases}
    k_{1d}x_1^{d-1}=k_{id}x_i^{d-1},\text{ for }i=2,\ldots,n,\\
\sum_{i=1}^n x_i=1.
\end{cases}
\]
From the first $(n-1)$ equations, we get that
\[x_i=\zeta_i\sqrt[d-1]{\frac{k_{1d}}{k_{id}}} x_1,\]
for some $\zeta_i$ a $(d-1)$-th root of unity, for all $i=2,n$. Now, for all the $(d-1)^{n-1}$ choices of the $\zeta_i$'s we plug in the expressions of the $x_i$'s to the last equation, and we get
\[x_1\left(1+\sum_{i=2}^n \zeta_i\sqrt[d-1]{\frac{k_{1d}}{k_{id}}}\right)=1,\] which has a unique solution in $x_1$. So this particular system has indeed $(d-1)^{n-1}$ many complex solutions. 

The generic degree is the maximal possible degree among all choices of $w_i$. Picking particular weights can only decrease the degree of the system. Now we have seen that the maximal degree is greater than or equal to the degree of this specific system with $w_1=\ldots=w_{d-1}=0$ and $w_d\neq 0$, which was $(d-1)^{n-1}$, but also by B\'ezout's theorem states that the generic degree is less than or equal to $(d-1)^{n-1}$, so we conclude that under the circumstances of the statement of the theorem, the degree of this system generically is exactly $(d-1)^{n-1}$.
\end{proof}

\begin{remark} We make the following observations.
\begin{itemize}
\item In the case of $d=2$, we get back the well-known fact that when only considering the first two cumulants (the expected value and the variance), we get one critical point of the system.
\item In the case of $n=2$, we will always get $d-1$ complex solutions (roots of a degree $(d-1)$ polynomial).
\item If $d$ is even, then we always have at least one real solution, because $(d-1)^{n-1}$ is odd, and complex solutions come in conjugate pairs.
\end{itemize}
\end{remark}

Observe that if $d$ is even and the leading term of $L$, that is $w_d\sum_{i=1}^n k_{id} x_i^{d}$, is negative definite, then we get a global maximum. Otherwise, there is no global maximum.
\begin{observation}
The generalized utility function $L$ has a global maximum if and only if $d$ is even and all the $d^{th}$ order cumulants are of the same sign, and $w_d$ has the opposite sign. Only in this case are the algebraic solutions financially valid. 
\end{observation}
\begin{example}
Let us now consider a portfolio consisting of two assets, and assume that the utility function takes into account the first four cumulants of the returns. In this case, our statistical data matrix $k=(k_{ij})$ is a $2\times 4$ matrix of real non-zero numbers. In this case, the system of critical equations looks like
\[
 \begin{cases}
\sum_{j=1}^4 j w_j k_{1j} x_1^{j-1}=\lambda,\\
\sum_{j=1}^4 j w_j k_{2j} x_2^{j-1}=\lambda,\\
x_1+x_2=1.
\end{cases}   
\]
One can find out, for example, by using Macaulay $2$ (see \cite{M2}), that before fixing the $k_{ij}$ and the $w_j$, this system gives rise to a degree $15$, codimension $3$ variety in the $8+4+2+1=15$ dimensional space of the $k_{ij}$'s, $w_j$'s, $x_i$'s and $\lambda$. For fixed $w_j$, we get a degree $7$, codimension $7$ variety. For fixed $k_{ij}$, this system gives rise to a degree $6$, codimension $11$ variety. Finally, when fixing both the $k_{ij}$'s and the $w_j$'s, we get a degree $3=(4-1)^{(2-1)}$, zero-dimensional variety. Moreover, the leading form of the utility function $L$ is equal to $w_4(k_{14}x_1^4+k_{24}x_2^4)$. If this is negative definite, then we will have a global maximum of $L$.
\end{example}
\subsection{Critical points at infinity}
We have seen in the previous section that under mild assumptions on $k_{ij}$-s, the utility function $L$ will always have $(d-1)^{n-1}$ complex critical points given that $w_d\neq 0$. In this section, we study those $(k_{ij})$-s for which we will have fewer than $(d-1)^{n-1}$ critical points even if we assume that $w_d\neq 0$. Actually, we do not “lose” critical points; they merge or escape to infinity, so the number counted with multiplicity remains constant.

The Bézout bound we used in the proof of Theorem~\ref{numberofcriticalpoints} originally holds for a system of homogeneous polynomials and acts as an upper bound for the non-homogeneous case. This is because of the following. If given a non-homogeneous system, we homogenize it using an extra variable $y$, and we get a homogeneous system with as many solutions as the product of the degrees. Now, each solution of the homogenized system is also a solution of the original non-homogeneous system by scaling it to get $y=1$, unless the homogeneous solution is of the form such that $y=0$. We call such a solution \textit{a solution at infinity}.

Now, we analyze when our system \ref{criticalequ} of critical equations loses solutions at infinity. By homogenizing the system \ref{criticalequ} using the variable $y$ we get
\[
 \begin{cases}
\sum_{j=1}^d j w_j k_{ij} x_i^{j-1}y^{d-j}=\lambda y^{d-1},\text{ for }i=1,\ldots,n,\\
\sum_{i=1}^n x_i=y.
\end{cases} 
\]
Now, setting $y=0$, we get
\[
\begin{cases}
w_d k_{id} x_i^{d-1}=0,\text{ for }i=1,\ldots,n,\\
\sum_{i=1}^n x_i=0.
\end{cases} 
\]
Now, since $w_d\neq 0$ and $k_{id}\neq 0$, we get that the only solution at infinity is when all $x_i$ are zero.
\section{The discriminant locus}
We have established so far that under the assumptions we made, the number of complex critical points of $L$ is constant, counting multiplicities, and equals $(d-1)^{n-1}$. Moreover, critical points never get lost to infinity. 
In this section, we try to understand the locus of those $w_i$-s, for which at least one of the critical points of $L$ gains multiplicity. This locus is the \textit{discriminant locus} of $L$.

Let us denote by $\mathcal{C}_{w,k}\subseteq \mathbb{C}^{n+1}_{\lambda,x_1\ldots,x_n}$ the zero-dimensional variety of critical points defined by the equations in the system \ref{criticalequ}, which was
\begin{equation}
\begin{cases}
P_i(x_i)-\lambda=0,\text{ for }i=1,\ldots,n,\\
\sum_{i=1}^n x_i-1=0,
\end{cases}   
\end{equation}
where $P_i(x_i)=\sum_{j=1}^d j w_j k_{ij} x_i^{j-1}$.
We aim to describe conditions on $\mathcal{C}_{w,k}$ so that the Jacobian of the system has a rank deficit. That is, conditions on the $w_j$'s and the $k_{ij}$, so that $\mathcal{C}_{w,k}$ has singular points. First, we differentiate the polynomials of the above system with respect to $\lambda$, then with respect to $x_1, x_2,\ldots,x_n$, and we see that the Jacobian is equal to
\[
\begin{bmatrix}
-1&\frac{\partial P_1}{\partial x_1}&0&\ldots&0\\
-1&0&\frac{\partial P_2}{\partial x_2}&\ldots&0\\
\vdots&\vdots&\vdots&\ddots&\vdots\\
-1&0&0&\ldots&\frac{\partial P_n}{\partial x_n}\\
0&1&1&\ldots&1
\end{bmatrix}.
\]
It has a rank deficit if it has a non-trivial kernel, so we need $v=(v_0,v_2,\ldots,v_{n})\in\mathbb{C}^{n+1}$ non-zero, such that
\[\mathrm{Jac}(\mathcal{C}_{w,k})\cdot v=0.\]
This translates to 
\[
\begin{cases}
\frac{\partial P_i(x_i)}{\partial x_i}v_i=v_0,\text{ for all }i=1,n,\\
\sum_{i=1}^n v_i=0.
\end{cases}
\]
There are three cases to consider. First, if there are at least two partial derivatives that are zero, say $\frac{\partial P_i(x_i)}{\partial x_i}=\frac{\partial P_j(x_j)}{\partial x_j}=0$, then setting $v_0=v_k=0$ for $k\neq i,j$ and $v_j=v_i\neq 0$ gives infinitely many non-zero vectors in the kernel. Now, suppose that exactly one of the $\frac{\partial P_i(x_i)}{\partial x_i}$-s is zero; then $v_0$ must be zero, and as a consequence, all the $v_j$-s must be zero, with $j\neq i$. But then the last condition forces $v_i$ to become zero as well, thus producing the trivial solution. Lastly, suppose that none of the partial derivatives is zero. Then, to get a non-trivial solution, we must have that no $v_i$ is zero (because if one of the $v_i$'s is zero, this forces $v_0$ to be zero, but then all the other $ v_i$'s become zero as well). We get that
\[
\begin{cases}
v_i=\frac{v_0}{\frac{\partial P_i(x_i)}{\partial x_i}},\text{ for all }i=1,n,\\
\sum_{i=1}^n v_i=0.
\end{cases}
\]
And plugging in the first $n$ equations into the last one and dividing by $v_0$, we get
\[
\sum_{i=1}^n \frac{1}{\frac{\partial P_i(x_i)}{\partial x_i}}=0.
\]
Now multiplying this equation by $\prod_{i=1}^n \frac{\partial P_i(x_i)}{\partial x_i}$, we get
\[
\sum_{i=1}^n \prod_{j\neq i}\frac{\partial P_j(x_j)}{\partial x_j}=0.
\]
To summarize this, we have the following theorem.
\begin{theorem}\label{thmdisc}
A critical point, $(x_1,x_2,\ldots,x_n)$, of the utility function $L$ has a multiplicity if and only if it satisfies
\begin{equation}\label{discriminant}
    \sum_{i=1}^n \prod_{j\neq i}\frac{\partial P_j(x_j)}{\partial x_j}=0,
\end{equation}
where $P_i(x_i)=\sum_{j=1}^d j w_j k_{ij} x_i^{j-1}$.
\end{theorem}
\begin{proof}
The proof of this theorem follows from the discussion preceding it; the only thought to be added is that in the case when two partials are zero, say $\frac{\partial P_i(x_i)}{\partial x_i}=\frac{\partial P_j(x_j)}{\partial x_j}=0$, the equation~\ref{discriminant} is satisfied automatically.
\end{proof}
Now, to get the locus of those $k_{ij}$ and $w_j$ for which $L$ has a critical point with a multiplicity, namely the discriminant locus of $L$, one has to eliminate the $x_i$ from the equation \ref{discriminant}. Let us see this in the following example.
\begin{example}
Let us consider the portfolio consisting of two assets, and we want to optimize an order four utility function $L$. The system of equations defining the critical points of $L$ is the following,
\[
 \begin{cases}
P_1(x_1)=\lambda,\\
P_2(x_2)=\lambda,\\
x_1+x_2=1,
\end{cases}   
\] where $P_i(x_i)=\sum_{j=1}^4 j w_j k_{ij} x_i^{j-1}$. Now, based on Theorem\ref{thmdisc} to have critical points with multiplicity, we need in addition that
\[
\frac{\partial P_1}{\partial x_1}+\frac{\partial P_2}{\partial x_2}=0.
\]
Adding this to the system and eliminating first $\lambda$, then $x_1$, and $x_2$ we get the following irreducible polynomial in variables $k_{ij}$ and $w_j$, of degree $8$, with $128$ terms
\[
9k_{12}^2k_{13}^2w_2^2w_3^2+18k_{12}k_{13}^2k_{22}w_2^2w_3^2+\ldots-2592k_{13}k_{14}k_{24}^2w_3w_4^3-1728k_{14}^2k_{24}^2w_4^4,
\] defining the discriminant locus of $L$.
\end{example}
\section{The feasible portfolio variety}
So far, we have approached portfolio optimization by optimizing the utility function $L=\sum_{j=1}^d w_j \kappa_j$, which incorporates the cumulants of the portfolio $P=\sum_{i=1}^n x_i X_i$, with respect to the constraint that $\sum_{i=1}^nx_i=1$. A different way is to first consider the so-called feasible portfolio set of points with coordinates equal to the possible cumulants of the portfolio depending on the choice of the $x_i$'s, and optimize the linear function $\sum_{j=1}^d w_j y_j$ on it.

Remember that the $j^{th}$ cumulant of our portfolio was
\[\kappa_j(P)=\sum_{i=1}^n x_i^j k_{ij},\] where we denoted $k_{ij}=\kappa_{j}(X_i)$. Then we have the following definition.
\begin{definition}
We call the \textbf{feasible portfolio variety} $\mathcal{P}_{k}$ the closure of the image of the hyperplane $\sum_{i=1}^n x_i=1$ under the map
\[
\begin{cases}
\phi:\mathbb{C}^n\to \mathbb{C}^d,\\
(x_1,x_2\ldots,x_n)\mapsto (\sum_{i=1}^n x_i^jk_{ij})_{j=1,d}.
\end{cases}
\] 
In order to incorporate that $\mathcal{P}_{k}$ is the image of the hyperplane $\sum_{i=1}^n x_i=1$ under $\phi$, we instead look at $\mathcal{P}_k$ as the closure of the image of the map
\[
\begin{cases}
\phi':\mathbb{C}^{n-1} \to \mathbb{C}^d,\\
(x_1,x_2,\ldots,x_{n-1})\mapsto (\sum_{i=1}^{n-1} x_i^j k_{ij}+(1-\sum_{i=1}^{n-1} x_i)^j k_{nj})_{j=1,d}.
\end{cases}
\]
\end{definition}
\begin{remark}
Specialized linear sections of the feasible portfolio variety (e.g., by setting certain cumulants to fixed values) represent the higher-order generalization of the so-called \textbf{efficient frontiers} \cite{Mar}.
\end{remark}
\begin{example}
A plot of the feasible portfolio curve (dimension $n-1=2-1=1$) in the case of two assets ($n=2$) of order three ($d=3$), for the following parametrization of $\mathcal{P}_k$
\[
x_1\mapsto (2 x_1+5\left(1-x_1\right),x_1^{2}+2\left(1-x_1\right)^{2},7x_1^{3}+\left(1-x_1\right)^{3}),
\]
is displayed on the left of Figure~\ref{portfoliofig}. In this case, $\mathcal{P}_k$ is of degree $3=d\cdot \ldots\cdot (d-n+2)$, and is generated by the vanishing of the polynomials
\[
\begin{cases}
 36y_2^2+18y_1y_3+367y_1-561y_2+81y_3-1028=0,\\
6y_1y_2+23y_1-63y_2+9y_3-58=0,\\
y_1^2-6y_1-3y_2+11=0.   
\end{cases}
\]

A plot of the feasible portfolio surface (dimension $n-1=3-1=2$) in the case of three assets ($n=3$) of order three ($d=3$), for the following parametrization of $\mathcal{P}_k$
\[
(x_1,x_2)\mapsto (2x_1-3x_2+5(1-x_1-x_2), x_1^2+5x_2^2-7(1-x_1-x_2)^2, -3x_1^3+x_2^3+(1-x_1-x_2)^3),\]
is displayed on the right of Figure~\ref{portfoliofig}. In this case, $\mathcal{P}_k$ is of degree $6=3\cdot 2=d\cdot \ldots\cdot (d-n+2)$ and it is generated by the vanishing of the polynomial
\[
1468y_1^6-5280y_1^5+46512y_1^4y_2+229584y_1^4+\ldots+6234864y_2+1303344y_3-2752336.
\]
\begin{figure}
    \centering
   \includegraphics[width=0.4\textwidth]{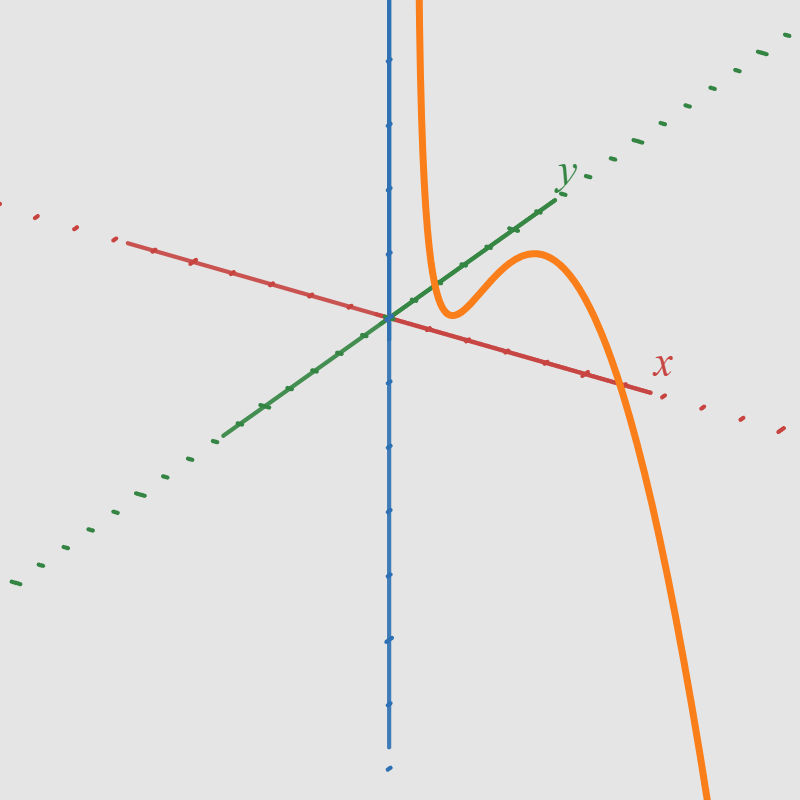} \hspace{1em}\includegraphics[width=0.55\textwidth]{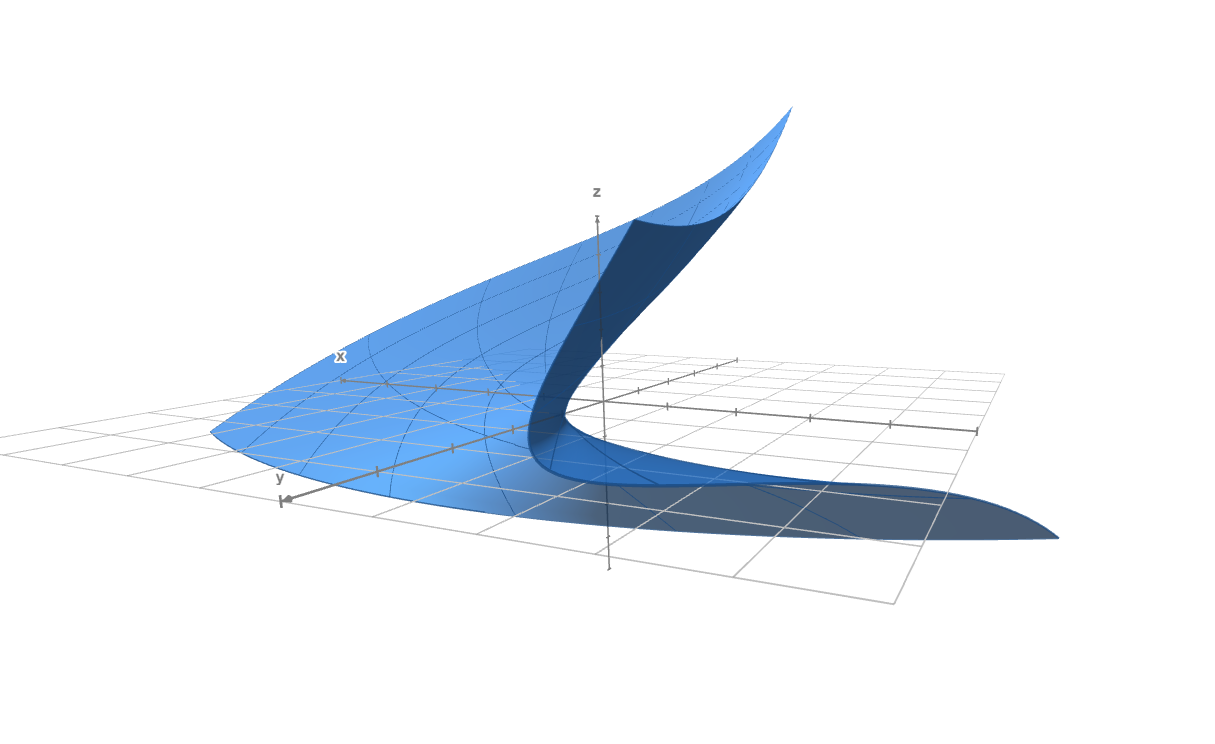}
    \caption{On the left, a two-asset order-three feasible portfolio curve. On the right, a three-asset order-three feasible portfolio surface.}
    \label{portfoliofig}
\end{figure}
We used Macaulay $2$ (see \cite{M2}) for the computations and Desmos (\url{https://www.desmos.com}) for the pictures.
\end{example}
We want to study the basic geometric properties (dimension and degree) of $\mathcal{P}_k$, for generic parameters $k_{ij}$.
\subsection{The dimension of $\mathcal{P}_k$}
We suppose that $d\geq n-1$. We want to show that for generic choices of $k_{ij}$, the image of the $ n-1$-dimensional hypersurface defined by $\sum_{i=1}^nx_i=1$ will be $n-1$ dimensional inside $\mathbb{C}^d$. 
For this, we pick a generic point on $\mathcal{P}_k$ parametrized by a generic choice of $(x_1, x_2,\ldots, x_n)$ and we compute the Jacobian of the map at this point and show that it is of full rank.

For this, we compute the derivatives, and we get that for $i=1,\ldots,n-1$
\[
\frac{\partial \phi'}{\partial x_i}=\left(jx_i^{j-1}k_{ij}-j\left(1-\sum_{i=1}^{n-1}x_i\right)^{j-1}k_{nj}\right)_{j=1,d}.
\]
If we denote $1-\sum_{i=1}^{n-1}x_i$ by $x_n$, then we get that
\[
\frac{\partial \phi'}{\partial x_i}=j\left(x_i^{j-1}k_{ij}-x_n^{j-1}k_{nj}\right)_{j=1,d}.
\]
Now, to see that $\mathrm{Jac}(\phi')$ is generically of full rank, we need one specific instance of the $k_{ij}$'s and the $x_i$'s so that we get a full rank Jacobian (this is to see that the variety of $k$'s and $x$'s for which all the maximal minors vanish is a proper variety of codimension at least one).

First, choose $x_1,x_2,\ldots,x_{n-1}$, such that $x_n=1-\sum_{i=1}^{n-1}x_i=0$ and $k_{n1}=0$ (this choice violates the non-zero assumption but is used only to show the existence of one full-rank instance), so get the row vectors
\[
\frac{\partial \phi'}{\partial x_i}=j\left(x_i^{j-1}k_{ij}\right)_{j=1,d}, \text{ for }i=1,n-1.
\]
These row vectors are independent if and only if the following row vectors are independent
\[\left(x_i^{j-1}k_{ij}\right)_{j=1,d}, \text{ for }i=1,n-1.\]
Now, when choosing all $k_{ij}=1$, we get a $(n-1)\times d$ Vandermonde matrix, which is of full rank if all the $x_i$'s are distinct. This can be easily realized. To summarize this we get the following statement.
\begin{proposition}\label{dimofPk}
In the case when $d\geq n-1$ and for generic values of the cumulants $k=(k_{ij})_{ij}$, the feasible portfolio variety $\mathcal{P}_k$ is $(n-1)$ dimensional.
\end{proposition}
\subsection{The degree of $\mathcal{P}_k$}
In this subsection, we further restrict and assume $d> n$. To get the degree of $\mathcal{P}_k$, we will have to intersect it with a generic complementary dimensional (by Proposition\ref{dimofPk} that is equal to $d-(n-1)$) hypersurface and count the number of intersection points.

Let $\mathcal{H}\subseteq \mathbb{C}^d$ be a $d-(n-1)$ dimensional generic hypersurface defined by the vanishing of
\[
\begin{cases}
\sum_{j=1}^d c_{1,j}y_j=c_{1,0},\\
\sum_{j=1}^d c_{2,j}y_j=c_{2,0},\\
\vdots\\
\sum_{j=1}^d c_{n-1,j}y_j=c_{n-1,0}.
\end{cases}
\]
Recall that $\mathcal{P}_k$ is the closure of the image of the hyperplane, defined by $\sum_{i=1}^n x_i=1$, under the map
\[
\begin{cases}
\phi:\mathbb{C}^{n} \to \mathbb{C}^d,\\
(x_1,x_2,\ldots,x_{n})\mapsto (\sum_{i=1}^{n}  k_{ij}x_i^j)_{j=1,d}.
\end{cases}
\]
Now, $\mathcal{H} \cap \mathcal{P}_k$ is given by 

\begin{equation}\label{degree_equ}
\begin{cases}
\sum_{j=1}^d c_{1,j}\left(\sum_{i=1}^{n} k_{ij}x_i^j\right)=c_{1,0},\\
\sum_{j=1}^d c_{2,j}\left(\sum_{i=1}^{n} k_{ij}x_i^j \right)=c_{2,0},\\
\vdots\\
\sum_{j=1}^d c_{n-1,j}\left(\sum_{i=1}^{n} k_{ij}x_i^j\right)=c_{n-1,0},\\
\sum_{i=1}^n x_i=1.
\end{cases}
\end{equation}

After homogenizing using $y$ we get 
\[
\begin{cases}
\sum_{j=1}^d c_{1,j}\left(\sum_{i=1}^{n} k_{ij}x_i^jy^{d-j}\right)=c_{1,0}y^d,\\
\sum_{j=1}^d c_{2,j}\left(\sum_{i=1}^{n} k_{ij}x_i^jy^{d-j} \right)=c_{2,0}y^d,\\
\vdots\\
\sum_{j=1}^d c_{n-1,j}\left(\sum_{i=1}^{n} k_{ij}x_i^jy^{d-j}\right)=c_{n-1,0}y^d,\\
\sum_{i=1}^n x_i=y.
\end{cases}
\]
At infinity, this system has a positive-dimensional solution set, namely the solution set of
\[
\begin{cases}
\sum_{i=1}^{n} k_{id}x_i^d=0,\\
\sum_{i=1}^n x_i=0.
\end{cases}
\]
Therefore, we must rewrite system \ref{degree_equ} in an equivalent form to prevent this from happening.
For $l=1,\ldots,n-1$ let us denote
\[
P_l(x_1,\ldots,x_{n})=c_{l,d}\sum_{i=1}^nk_{i,d}x_i^{d}+c_{l,d-1}\sum_{i=1}^nk_{i,d-1}x_i^{d-1}+\ldots+c_{l,1}\sum_{i=1}^nk_{i,1}x_i-c_{l,0}.
\]
We can observe now that all degree $k$ terms of the $P_l$'s are scalar multiples of each other, for all $k=0,\ldots,d$. By multiplying with non-zero scalars and adding the equations of the system~\ref{degree_equ}, we do not change the solution set of it, so we get that system~\ref{degree_equ} is equivalent to
\[
\begin{cases}
P_1=0,\\ 
Q_2:=c_{1d}P_2-c_{2,d}P_1=(c_{1,d}c_{2,d-1}-c_{2,d}c_{1,d-1})\sum_{i=1}^nk_{i,d-1}x_i^{d-1}+l.o.t=0,\\
\vdots\\
Q_{n-1}:=c_{1,d}P_{n-1}-c_{n-1,d}P_1=(c_{1,d}c_{n-1,d-1}-c_{n-1,d}c_{1,d-1})\sum_{i=1}^nk_{i,d-1}x_i^{d-1}+l.o.t=0,\\
\sum_{i=1}^n x_i=1.
\end{cases}
\]
Observe that here again all degree $k$ terms of the $Q_l$'s are scalar multiples of each other, for all $k=0,\ldots,d-1$, so we repeat the same process.

By repeating this, we finally get that system~\ref {degree_equ} is equivalent to 
\[
\begin{cases}
m_{1,d}(\sum_{i=1}^nk_{i,d}x_i^{d})+m_{1,d-1}(\sum_{i=1}^nk_{i,d-1}x_i^{d-1})+\text{ lower order terms}=0,\\
m_{2,d-1}(\sum_{i=1}^nk_{i,d-1}x_i^{d-1})+m_{2,d-2}(\sum_{i=1}^nk_{i,d-2}x_i^{d-2})+\text{ lower order terms}=0,\\
\vdots\\
m_{n-1,d-n+2}(\sum_{i=1}^nk_{i,d-n+2}x_i^{d-n+2})+m_{n-1,d-n+1}(\sum_{i=1}^nk_{i,d-n+1}x_i^{d-n+1})+\ldots=0,\\

\sum_{i=1}^n x_i=1,
\end{cases}
\]
with some corresponding generic coefficients $m_{l,k}\in \mathbb{R}$. 
Now, by Bezout’s theorem, we get at most $d\cdot(d-1) \cdot \ldots\cdot(d-n+2)\cdot 1$
solutions.
And this is the generic degree of this system because this bound can be reached for
$m_{i,d+1-i}=1$, $m_{i,j}=0$, $m_{i,0}=-1$ for $i=1,d-n+2$ and $j>d+1-i$, and for $k_{i,d+1-i}=1$, $k_{i,j}=0$, for $i=1,d-n+2$ and $j\neq d+1-i$. In this case, we get the system to be 
\[
\begin{cases}
x_1^d=1,\\ 
x_2^{d-1}=1,\\
\vdots \\
x_{n-1}^{d-n+2}=1,\\
\sum_{i=1}^n x_i=1,
\end{cases}
\]
From the first equation, $x_1$ can be any of the $d^{th}$-roots of unity; from the second equation, $x_2$ can be any of the $(d-1)^{th}$-roots of unity; etc. Finally, the last equation determines $x_n$ uniquely.
This concludes the count of roots of System~\ref {degree_equ} for generic coefficients. 

To conclude our finding, the only remaining thing to see is that in the case when $d> n$ the map $\phi$ is birational onto its image. Indeed, the image $Y = \overline{\phi(\mathbb{C}^n)} \subset \mathbb{C}^d$ is an $n$-dimensional affine variety. Let $y = \phi(x^*) \in Y$ be generic. Then the first $n$ equations $\sum_{i=1}^n k_{ij} x_i^j = y_j$ for $j = 1, n$ form an $n \times n$ polynomial system with at most $n!$ solutions in $\mathbb{C}^n$, one of which is $x^*$. The remaining $d - n \ge 1$ equations $\sum_{i=1}^n k_{ij} x_i^j = y_j$ for $j = n+1, d$ impose additional generic constraints on the $n! - 1$ secondary solutions. Because the coefficients $k_{i, j}$ are chosen generically, none of the $n! - 1$ secondary solutions satisfy these additional equations. Hence, $x^*$ is the unique preimage in $\mathbb{C}^n$ for a generic point $y \in Y$.

To summarize, we have the following result.
\begin{theorem}\label{degree_of_Pk}
In the case when $d>n$, for generic values of the cumulants $k=(k_{ij})_{ij}$, the degree of the feasible portfolio variety $\mathcal{P}_k$ is \[d\cdot(d-1)\cdot\ldots\cdot(d-n+2).\]
\end{theorem}
\begin{remark}[\textbf{Further directions}]
Beyond the necessity of testing these generalized portfolios to concrete real-world data, further theoretical aspects should be studied. Some of these are the study of more complex geometric and topological features of the feasible portfolio variety; understanding the special portfolios lying on the discriminant; results on the number of real solutions -- in view of the connected components of the complement of the discriminant; existence and number of positive solutions; understanding solutions in $(0,1)$; etc.
\end{remark}

\paragraph{\textbf{Acknowledgements}.} The author was supported by the project “Singularities and Applications” - CF 132/31.07.2023 funded by the European Union - NextGenerationEU - through Romania’s National Recovery and Resilience Plan and by the Domus "Alkotói" Scholarship of the Hungarian Academy of Science.

\end{document}